\documentclass[11pt]{article}
\usepackage[utf8]{inputenc}
\usepackage{style}
\usepackage{framed}
\usepackage[parfill]{parskip}
\usepackage{setspace}
\usepackage{hyperref}

\allowdisplaybreaks

\title{A Note on Degree vs Gap of Min-Rep Label Cover and Improved Inapproximability for Connectivity Problems}

\author{
Pasin Manurangsi\thanks{Email: \texttt{pasin@berkeley.edu}. Supported by NSF under Grants No. CCF 1655215 and CCF 1815434.} \vspace{-0.5em}\\
UC Berkeley
}

\begin{document}

\maketitle

\begin{abstract}
This note concerns the trade-off between the degree of the constraint graph and the gap in hardness of approximating the Min-Rep variant of Label Cover (aka Projection Game). We make a very simple observation that, for NP-hardness with gap $g$, the degree can be made as small as $O(g \log g)$, which improves upon the previous $\tilde{O}(g^{1/2})$ bound from a work of Laekhanukit~\cite{Bun14}. Note that our bound is optimal up to a logarithmic factor since there is a trivial $\Delta$-approximation for Min-Rep where $\Delta$ is the maximum degree of the constraint graph.

Thanks to known reductions~\cite{CLNV14,CK12,CMVZ16,Bun14}, this improvement implies better hardness of approximation results for Rooted $k$-Connectivity, Vertex-Connectivity Survivable Network Design and Vertex-Connectivity $k$-Route Cut. 
\end{abstract}

\section{Introduction}

We study variants of the Label Cover problem. The input to these problems is a \emph{label cover instance} $\Pi$, which consists of (i) a bipartite graph $G = (A, B, E)$ called the \emph{constraint graph}, (ii) an \emph{alphabet set (aka label set)} $\Sigma$, and, (iii) for each edge $(a, b) \in E$, a (projection) \emph{constraint} $\pi_{(a, b)}: \Sigma \to \Sigma$ which can be any function from $\Sigma$ to itself.

A \emph{labeling} $\phi: (A \cup B) \to \Sigma$ is a function that assigns each vertex a label from $\Sigma$. The value of a labeling $\phi$, denoted by $\val_{\Pi}(\phi)$, is the fraction of edges $(a, b) \in E$ such that $\pi_{(a, b)}(\phi(a)) = \phi(b)$. These edges (or constraints) are said to be \emph{satisfied} by $\phi$.

A \emph{multilabeling} $\psi: (A \cup B) \to \cP(\Sigma)$ is a generalization of labeling in which each vertex can be now assigned multiple labels. (Here $\cP(\Sigma)$ denote the power set of $\Sigma$.) The value of a multilabeling $\psi$, denoted by $\val_{\Pi}(\psi)$, is the fraction of edges $(a, b) \in E$ such that $\pi_{(a, b)}(\psi(a)) \cap \psi(b) \ne \emptyset$ (i.e. there exists $\sigma_a \in \psi(a)$ and $\sigma_b \in \psi(b)$ such that $\pi_{(a, b)}(\sigma_a) = \sigma_b$). Similar to before, such edges are said to be satisfied by $\psi$. The \emph{cost} of a multilabeling $\psi$, denoted by $\cost(\psi)$, is simply $\sum_{v \in A \cup B} |\psi(v)|$.

The two variants of Label Cover we focus on are the \emph{Max-Rep} and \emph{Min-Rep} variants. The Max-Rep variant asks for a labeling with maximum value, whereas the Min-Rep variant asks for a multilabeling that satisfies all constraints and has minimum cost. We use $\val(\Pi)$ and $\minrep(\Pi)$ to denote the optimum of $\Pi$ for Max-Rep and Min-Rep respectively, i.e., $\val(\Pi) := \max_{\phi: (A \cup B) \to \Sigma} \val(\phi)$ and $\minrep(\Pi) = \min_{\psi: (A \cup B) \to \cP(\Sigma) \text{ such that } \val(\psi) = 1}\cost(\psi)$.

Both Min-Rep and Max-Rep problems are standard starting points used in numerous hardness of approximation reductions. In some applications, the parameters one get in inapproximability results for Label Cover affect the resulting inapproximability factors of the reductions. The two parameters of interest in this note are the alphabet size $|\Sigma|$ and the maximum degree of the constraint graph $G$. We denote the latter by $\Delta(G)$, or simply $\Delta$ when $G$ is clear from context.

Suppose that we want a gap of $g > 1$ in our NP-hardness of approximation of {\em Max-Rep}. Raz's parallel repetition theorem~\cite{Raz98} combined with the PCP Theorem~\cite{AS,ALMSS} implies that $|\Sigma| = O(g^{p})$ suffices where $p$ is some (large) constant. This was later improved by Khot and Safra~\cite{KS13} to $|\Sigma| = O(g^6)$ and later by Chan~\cite{Chan16} to $|\Sigma| = O(g^2 \log g)$. On the other hand, standard sparsification arguments imply that the maximum degree can be made as small as $O(g \log g)$. (See e.g.~\cite{Bun14} where such an argument is carried out in full details.)

For {\em Min-Rep}, the typical way to show hardness of approximation with factor $g$ is via proving hardness of approximating of {\em Max-Rep} of factor $O(g^2)$. This was indeed also the route taken in~\cite{Bun14}; due to the square loss in parameter from Max-Rep, for a gap of $g$ in Min-Rep, the alphabet size and the maximum degree now become $O(g^4 \log^2 g)$ and $O(g^2 \log^2 g)$ respectively.

\subsection{Our Results}

Our main result is an improvement for the degree parameter. In particular, we show that the degree can be made as small as $O(g \log g)$ (instead of $\tilde{O}(g^2)$ in~\cite{Bun14}) as stated below.

\begin{theorem} \label{thm:main}
For every positive integer $g > 1$, it is NP-hard (under randomized reduction) to, given a Label Cover instance $\Pi$ of alphabet size $O(g^4 \log^2 g)$ and maximum degree $O(g \log g)$, approximate $\minrep(\Pi)$ to within a factor of $g$.
\end{theorem}

We remark that our degree bound of $O(g \log g)$ here is optimal up to $O(\log g)$ factor; when the maximum degree is $\Delta$, the trivial algorithm that, for each edge $(a, b) \in E$, adds an arbitrary label $\sigma \in \Sigma$ to $a$ and the corresponding label $\pi_{(a, b)}(\sigma)$ to $b$ achieves a $\Delta$-approximation for Min-Rep.

Both our proof and Laekhanukit's use the same reduction: roughly speaking, we start with a hard instance of Max-Rep (that may not be sparse) and subsample a sparse subgraph of the constraint graph to create a new instance. The main (simple) observation that leads to the improvement is that we can analyze the soundness of Min-Rep of this new instance directly without going through Max-Rep; namely, by bounding the probability that each multilabeling satisfies all subsampled edges and use union bound, we arrive at the improvement. The argument is formalized in Section~\ref{sec:main-proof}.

\subsubsection{Improved Hardness for Connectivity Problems}

As mentioned earlier, Label Cover problems are the starting points of many hardness of approximation results. Specifically, the degree parameter affects the resulting inapproximation ratios for the following connectivity problems: Rooted $k$-Connectivity, Vertex-Connectivity Survivable Network Design and Vertex-Connectivity $k$-Route Cut. Before we state what our improvement implies for these problems, let us first give definitions of these problems and state some of the known approximation guarantees for the problems.

{\bf Rooted $k$-Connectivity (R$k$C) Problem.} In R$k$C, we are given an undirected graph $G$, a root vertex $r$ and a set of terminals $T \subseteq V$. The goal is to find a subgraph $H$ of $G$ with minimum number of edges such that, for every $t \in T$, there exist at least $k$ vertex-disjoint paths from $r$ to $t$ in the subgraph $H$. The problem has been extensively studied in literature~\cite{CCK08,CK08a,CK08,CK12,CK14,Nutov12,Nutov09,CLNV14}; the best known polynomial-time approximation algorithm due to Nutov~\cite{Nutov12} achieves an approximation ratio of $O(k \log k)$. 

{\bf Vertex-Connectivity Survivable Network Design (VC-SND) Problem.} In VC-SND, we are given an undirected graph $G$, a set $T \subseteq V$ of terminals and, for each pair $s, t \in T$ of terminals, a connectivity requirement $r_{st} \in \N \cup \{0\}$. The goal is to find a subgraph $H$ of $G$ with minimum number of vertices such that the requirements are satisfied, i.e., for every $s, t \in T$, there are at least $r_{st}$ vertex-disjoint paths from $s$ to $t$. The best known approximation algorithm for this problem by Chuzhoy and Khanna~\cite{CK12} achieves an $O(k^3\log |T|)$-approximation where $k$ denotes the maximum connectivity requirement (i.e. $k = \max_{s, t \in T} r_{st}$). 

{\bf Vertex-Connectivity $k$-Route Cut (VC-$k$RC) Problem.} Given an undirected graph $G$, $\cD$ source-sink pairs $(s_1, t_1), \dots, (s_{\cD}, t_{\cD})$ and an integer $k$, the goal is to find a smallest set $E'$ of edges such that, when edges in $E'$ are removed from $G$, there are less than $k$ vertex-disjoint paths from $s_i$ to $t_i$ for all $i \in [k]$. The best known true approximation algorithm of Chuzhoy \etal~\cite{CMVZ16} achieves an approximation ratio\footnote{It should be noted that, in~\cite{CMVZ16}, the authors did not explicitly state this; rather they gave an $O(k)$-approximation algorithm for the single source-sink pair case. It is clear that we can simply run this algorithm on each source-sink pair and take the union as the answer, which indeed yields an $O(\cD k)$-approximation.} $O(\cD k)$. We note here that Chuzhoy \etal~\cite{CMVZ16} also devise a bi-criteria approximation algorithm for the problem with a ratio that can be better than the aforementioned (true) approximation when no vertex participates in too many source-sink pairs; we choose not to state the ratio here, since our result does not apply to bi-criteria approximation algorithms.

The above three problems are shown to be hard to approximate to within a factor of $k^{\delta}$ for some small (implicit) constant $\delta > 0$ in~\cite{CLNV14,CK12,CMVZ16}. In \cite{Bun14}, Laekhanukit formulated the reduction in terms of Min-Rep Label Cover and observed that one can make $\delta$ explicit if the degree and the alphabet size can be made explicit polynomials of the gap. More precisely, the following theorem, which is a restatement of Theorem 3.1 in \cite{Bun14}, captures the dependency between the alphabet size, the degree and the resulting connectivity parameter.

\begin{theorem}[\cite{CLNV14,CK12,CMVZ16,Bun14}] \label{thm:red-param}
There exists an approximation-preserving reduction from Min-Rep Label Cover of maximum degree $\Delta$ and alphabet $\Sigma$ to R$k$C, VC-SND and VC-$k$RC with the following connectivity parameters: $k = O(\Delta^3 |\Sigma| + \Delta^4)$ for Rooted $k$-connectivity, $k = O(\Delta |\Sigma| + \Delta^2)$ for Vertex-Connectivity Survivable Network Design, and $k = O(\Delta |\Sigma|)$ for Vertex-Connectivity $k$-Route Cut.
\end{theorem}

As mentioned earlier, for NP-hardness of Min-Rep Label Cover with gap $g$, the result of~\cite{Bun14} gets the degree and the alphabet size to be as small as $\tilde{O}(g^2)$ and $\tilde{O}(g^4)$ respectively. With this, Laekhanukit applies Theorem~\ref{thm:red-param} which immediately gives NP-hardness of approximation with factors $\tilde{\Omega}(k^{1/10}), \tilde{\Omega}(k^{1/6})$ and $\tilde{\Omega}(k^{1/6})$ for Rooted $k$-connectivity, Vertex-Connectivity Survivable Network Design and Vertex-Connectivity $k$-Route Cut respectively. By starting with Theorem~\ref{thm:main} which has a better degree dependency than that of~\cite{Bun14}, we immediately arrive at the following improved hardness of approximation results for the three problems.

\begin{corollary}
For every sufficiently large connectivity parameter $k$, it is NP-hard to approximate the rooted $k$-connectivity problem to within $\tilde{\Omega}(k^{1/7})$ factor, the vertex-connectivity survivable network design problem to within $\tilde{\Omega}(k^{1/5})$ factor, and, the vertex-connectivity $k$-route cut problem to within $\tilde{\Omega}(k^{1/5})$ factor.
\end{corollary}

Another parameter considered in~\cite{Bun14} is the number of demand pairs $\cD$. Note that $\cD = |T|$ in R$k$C and $\cD$ is the number of $\{s, t\} \subseteq T$ such that $r_{st} > 0$ in VC-SND. Laekhanukit~\cite{Bun14} proved the following theorem:

\begin{theorem}[\cite{Bun14}] \label{thm:red-param-demand-pairs}
There exists an approximation-preserving reduction from Min-Rep Label Cover of maximum degree $\Delta$ to R$k$C, VC-SND and VC-$k$RC with $\cD = O(\Delta^2)$.
\end{theorem}

Plugging Theorem~\ref{thm:red-param-demand-pairs} to his Min-Rep hardness, Laekhanukit shows $\tilde{\Omega}(\cD^{1/4})$ ratio hardness for all three problems~\cite{Bun14}. By using our improved hardness for Min-Rep instead, we immediately arrive at an improved factor of $\tilde{\Omega}(\cD^{1/2})$:

\begin{corollary}
For every sufficiently large number of demand pairs $\cD$, it is NP-hard to approximate R$k$C, VC-SND and VC-$k$RC to within a factor of $\tilde{\Omega}(\cD^{1/2})$.
\end{corollary}

\section{Preliminaries}

For any graph $G$ and a vertex $u$ of $G$, we use $\Gamma_G(u)$ to denote the set of neighbors of $u$ in $G$.

For succinctness, we will say that a Label Cover instance $\Pi$ is regular or has maximum degree $\Delta$ as a shorthand for its associated constraint graph having such properties. To prove our result, we need the following hardness of approximation for Max-Rep Label Cover due to Chan~\cite{Chan16}:

\begin{theorem}[\cite{Chan16}] \label{thm:Chan}
For some constant $C > 0$, any prime power $q$ and any $\delta > 0$, it is NP-hard to, given a regular Label Cover instance $\Pi$ of alphabet size $q^2$, distinguish between $\val(\Pi) \geqs 1 - \delta$ and $\val(\Pi) < C\log q/q + \delta$. 
\end{theorem}

We note here that the instance produced in~\cite{Chan16} may not be regular. However, it is known that one can turn any Label Cover instance to a regular instance while approximately preserving the value of the instance. For instance, see Sections 4.2 and 4.3 of~\cite{Bun14}.

\section{Proof of the Main Theorem} \label{sec:main-proof}

In this section, we prove our main result (Theorem~\ref{thm:main}). Our main contribution is in fact a generic randomized sparsification lemma for Label Cover which is stated below.

\begin{lemma} \label{lem:main}
For any $\varepsilon, \gamma > 0$, there is a polynomial time randomized reduction that, given any regular Label Cover instance $\Pi = (G = (A, B, E), \Sigma, \{\pi_e\}_{e \in E})$ where $N = |A| + |B|$, outputs a Label Cover instance $\Pi'$ with alphabet set $\Sigma$ and maximum degree at most $\Delta := 10^6(2\log(2|\Sigma|)/\sqrt{\gamma})$ such that 
\begin{itemize}
\item (Completeness) if $\val(\Pi) \geqs 1 - \varepsilon$, then $\minrep(\Pi') \leqs (1 + \varepsilon\Delta)N$ with probability 0.9, and,
\item (Soundness) if $\val(\Pi) < \gamma$, then $\minrep(\Pi') > (0.06 / \sqrt{\gamma})N$ with probability 0.9.
\end{itemize}
\end{lemma}

Before we prove Lemma~\ref{lem:main}, we note that Theorem~\ref{thm:main} follows from it simply by plugging in appropriate parameters.

\begin{proof}[Proof of Theorem~\ref{thm:main}]
Let $C$ be the constant from Theorem~\ref{thm:Chan}. Let $q$ be the smallest prime power such that $q/\log q > 10^5 C g^2$, $\gamma$ be $2C\log q/q$, $\varepsilon$ be $1/\Delta$, and $\delta$ be $\min\{\varepsilon, Cq\log q\}$. Note that $q = \Theta(g^2 \log g)$ and $\gamma = 1/\Theta(g^2)$. From Theorem~\ref{thm:Chan}, it is hard to, given a regular Label Cover instance $\Pi$ with alphabet size $q^2$, distinguish between $\val(\Pi) \geqs 1 - \delta \geqs 1 - \varepsilon$ and $\val(\Pi) < C\log q/q + \delta < 2C\log q/q < 10^{-4}/g^2$, where the inequalities follow from our choices of parameters.

Applying the reduction in Lemma~\ref{lem:main} to $\Pi$, we arrive at an instance $\Pi'$ with alphabet size $q^2 = O(g^4 \log^2 g)$ and maximum degree at most $O(\log q/\sqrt{\gamma}) = O(g \log g)$. Moreover, the completeness guarantee of Lemma~\ref{lem:main} implies that, if $\val(\Pi) \geqs 1 - \varepsilon$, then $\minrep(\Pi') \leqs (1 + \varepsilon\Delta)N = 2N$ with probability 0.9. Similarly, the soundness guarantee of Lemma~\ref{lem:main} implies that, if $\val(\Pi) < 10^{-4}/g^2$, then $\minrep(\Pi') > (0.06/\sqrt{10^{-4}/g^2})N > 2gN$ with probability 0.9. Hence, it is NP-hard (under randomized reduction) to approximate $\minrep(\Pi')$ to within a factor of $g$. 
\end{proof}

The rest of this section is devoted to the proof of Lemma~\ref{lem:main} and is organized as follows. First, we describe the reduction and state a few useful facts in Section~\ref{sec:red}. The soundness and completeness of the reductions are then proved in Sections~\ref{sec:completeness} and~\ref{sec:soundness} respectively.

\subsection{The Reduction} \label{sec:red}

Once again, we remark here that the reduction is exactly the same as that in~\cite{Bun14}, except that the soundness analysis is different. The reduction in~\cite{Bun14} is in turn inspired by a reduction for independent set on bounded degree graphs by Austrin \etal~\cite{AKS11} who used the similar subsampling and high-degree vertex trimming techniques in their reduction; indeed our soundness analysis is more similar to~\cite{AKS11} than to~\cite{Bun14}. Nonetheless, we state the reduction in full here. Before we do so, let us first mention that we may assume that the degree $D$ of the constraint graph in the original instance $\Pi$ is arbitrarily large\footnote{This is because, for any positive integer $t$, we may create $t$ copies of the vertex sets $A, B$ and add constraints correspondingly (i.e. the new graph is $(A \times [t], B \times [t], E_t)$ where $E_t = \{((a, i), (b, j)) \mid (a, b) \in E, i, j \in [t]\}$ with $\pi_{((a, i), (b, j))} = \pi_{(a, b)}$). This ensures that the degree of every vertex is at least $t$ without effecting the value of the instance.}. In particular, we will assume throughout this section that $D \geqs 10000\Delta$ where $\Delta$ is the parameter in the statement of Lemma~\ref{lem:main}.

The reduction proceeds in two steps. First, we create an intermediate instance $\Pi_{\interm} = (G_{\interm} = (A, B, E_{\interm}), \Sigma, \{\pi_e\}_{e \in E_{\interm}})$ by simply include each edge $e \in E$ into $E_{\interm}$ independently at random with probability $p := 0.0001\Delta/D$ where $D$ denotes the degree of the original graph $G = (A, B, E)$. Then, we create the instance $\Pi' = (G' = (A, B, E'), \Sigma, \{\pi_e\}_{e \in E'})$ by removing every edge such that one of its endpoint has degree more than $\Delta$. More formally, let $A' = \{a \in A \mid \Gamma_{G_{\interm}}(a) \leqs \Delta\}$, $B' = \{b \in B \mid \Gamma_{G_{\interm}}(b) \leqs \Delta\}$ and $E' = (A' \times B') \cap E_{\interm}$. We note here that the vertex sets, the alphabet sets and constraints (for remaining edges) of $\Pi_{\interm}$ and $\Pi'$ remain the same from $\Pi$.

Throughout this section, we use $n$ to denote $|A| = |B| = N/2$. The following proposition states a fact that is useful in our proofs. Note that the probabilistic bound can be sharpen to $1 - o(1)$ but, since we do not need it here, we choose to state the simpler proof below.
\begin{proposition} \label{prop:simple-facts}
With probability 0.99, we have $|E_{\interm} \setminus E'| \leqs 0.1p\Delta n$.
\end{proposition}

\begin{proof}
By Markov inequality, it suffices to show that $\Ex[|E_{\interm} \setminus E'|] \leqs 0.001p\Delta n$. Due to linearity of expectation, $\Ex[|E_{\interm} \setminus E'|] = \sum_{e \in E} \Pr[e \in E_{\interm} \setminus E']$. For a fixed edge $e = (a, b) \in E$, we can bound $\Pr[(a, b) \in E_{\interm} \setminus E']$ as follows.
\begin{align}
\Pr[(a, b) \in E_{\interm} \setminus E'] 
&= \Pr[(a, b) \in E_{\interm}] \Pr[(a, b) \notin E' \mid (a, b) \in E_{\interm}] \nonumber \\
&= p \cdot \Pr[a \notin A' \vee b \notin B' \mid (a, b) \in E_{\interm}] \nonumber \\
&\leqs p \cdot \left(\Pr[a \notin A' \mid (a, b) \in E_{\interm}] + \Pr[b \notin B' \mid (a, b) \in E_{\interm}]\right) \label{eq:bound-removed-edge}
\end{align} 

Let us bound $\Pr[a \notin A' \mid (a, b) \in E_{\interm}]$. Notice that $a \notin A'$ if and only if $|\Gamma_{G_{\interm}}(a)| > \Delta$. Moreover, observe that $\Ex[|\Gamma_{G_{\interm}}(a)| \mid (a, b) \in E_{\interm}] = 1 + p(D - 1)$. Hence, by Markov inequality, we have
\begin{align*}
\Pr[|\Gamma_{G_{\interm}}(a)| > \Delta \mid (a, b) \in E_{\interm}] < \frac{1 + p(D - 1)}{\Delta} = 1/\Delta + pD/\Delta \leqs 0.0002.
\end{align*}
In other words, we have $\Pr[a \notin A' \mid (a, b) \in E_{\interm}] \leqs 0.0002$ and similarly $\Pr[b \notin B' \mid (a, b) \in E_{\interm}] \leqs 0.0002$. Plugging these back into~\eqref{eq:bound-removed-edge}, we have $\Pr[(a, b) \in E_{\interm} \setminus E'] \leqs 0.0004p$. By summing this over all edges $(a, b) \in E$, we have $\Ex[|E_{\interm} \setminus E'|] \leqs 0.0004p\Delta n < 0.001p\Delta n$. As already mentioned, Markov inequality then implies $\Pr[|E_{\interm} \setminus E'| > 0.1p\Delta n] < 0.01$ as desired.
\end{proof}

\subsection{Completeness Analysis} \label{sec:completeness}

The goal of this section is to prove the completeness part of Lemma~\ref{lem:main} as stated below.

\begin{lemma} \label{lem:completeness}
If $\val(\Pi) \geqs 1 - \varepsilon$, then, with probability 0.9, $\minrep(\Pi') \leqs (1 + \varepsilon\Delta)N$.
\end{lemma}

\begin{proof}
Suppose that $\val(\Pi) \geqs 1 - \varepsilon$; that is, there is a labeling $\phi: (A \cup B) \to \Sigma$ for $\Pi$ such that $\val_\Pi(\phi) \geqs 1 - \varepsilon$. Let $E_{\unsat} \subseteq E$ denote the set of (at most $\varepsilon |E|$) edges in $E$ that is not satisfied by $\phi$, i.e., $E_{\unsat} = \{(a, b) \in E \mid \pi_{(a, b)}(\phi(a)) \ne \phi(b)\}$. We define a multilabeling $\psi: (A \cup B) \to \cP(\Sigma)$ for $\Pi'$ as follows. First, start with $\psi(v) = \{\phi(v)\}$ for all $v \in A \cup B$. Then, for every $(a, b) \in E_{\unsat} \cap E'$, pick an arbitrary $\sigma \in \Sigma$, and add $\sigma$ to $\psi(a)$ and $\pi_{(a, b)}(\sigma)$ to $\psi(b)$.

It is clear that $\psi$ satisfies all edges in $E'$ and that $\cost(\psi)$ is at most $N + 2|E_{\unsat} \cap E'| \leqs N + 2|E_{\unsat} \cap E_{\interm}|$. Hence, it suffices to show that $\Pr[|E_{\unsat} \cap E_{\interm}| > 0.5\varepsilon\Delta N] \leqs 0.1$.

Observe that $\Ex[|E_{\unsat} \cap E_{\interm}|] = p |E_{\unsat}| \leqs p (\varepsilon |E|) = 0.00005 \varepsilon \Delta N$. As a result, Markov inequality implies that $\Pr[|E_{\unsat} \cap E'| > 0.5\varepsilon\Delta N] \leqs 0.0001 < 0.1$ as desired.
\end{proof}

\subsection{Soundness Analysis} \label{sec:soundness}

Finally, we will prove the soundness of the reduction, as stated below.

\begin{lemma} \label{lem:soundness}
If $\val(\Pi) < \gamma$, then $\minrep(\Gamma') > 0.06N/\sqrt{\gamma}$ with probability at least 0.9.
\end{lemma}

Before we prove Lemma~\ref{lem:soundness}, let us provide a couple useful propositions. The first states that any multilabeling of cost $\leqs 0.06N/\sqrt{\gamma}$ must leave more than half of the edges of $\Pi$ unsatisfied:

\begin{proposition} \label{prop:max-to-min}
If $\val(\Pi) < \gamma$, any multilabeling $\psi$ of $\Pi$ with $\cost(\psi) \leqs 0.06N/\sqrt{\gamma}$ has $\val_{\Pi}(\psi) < 0.5$.
\end{proposition}

\begin{proof}
Suppose for the sake of contradiction that $\val(\Pi) < \gamma$ but, for some multilabeling $\psi$ with $\cost(\psi) \leqs (0.06/\sqrt{\gamma})N$, we have $\val_{\Pi}(\psi) \geqs 0.5$. Consider a random labeling $\phi$ of $\Pi$ generated as follows: for every $v \in A \cup B$, if $\psi(v)$ is non-empty, then let $\phi(v)$ be a random element of $\psi(v)$. Otherwise, let $\phi(v)$ be an arbitrary label. We will show that $\Ex[\val_{\Pi}(\phi)] \geqs \gamma$; this implies that there exists a labeling with value at least $\gamma$ which is a contradiction to $\val(\Pi) < \gamma$.

To bound $\Ex[\val_{\Pi}(\phi)]$, let $E_{\sat}$ denote the set of all edges satisfied by $\psi$; that is, $E_{\sat} = \{(a, b) \in E \mid \exists \sigma_a \in \psi(a), \sigma_b \in \psi(b) \text{ such that } \pi_{(a, b)}(\sigma_a) = \sigma_b\}$. Since $\val_{\Pi}(\psi) \geqs 0.5$, we have $|E_{\sat}| \geqs 0.5Dn$.

Furthermore, let $A_{\leqs 0.3/\sqrt{\gamma}}$ and $B_{\leqs 0.3/\sqrt{\gamma}}$ be the sets of vertices in $A$ and $B$ respectively such that $\psi$ assigns at most $0.3/\sqrt{\gamma}$ labels to it, i.e., $A_{\leqs 0.3/\sqrt{\gamma}} = \{a \in A \mid |\psi(a)| \leqs 0.3/\sqrt{\gamma}\}$ and $B_{\leqs 0.3/\sqrt{\gamma}} = \{b \in B \mid |\psi(b)| \leqs 0.3/\sqrt{\gamma}\}$. Since $\cost(\psi) \leqs (0.06/\sqrt{\gamma})N$, we can conclude that $|A \setminus A_{\leqs 0.3/\sqrt{\gamma}}| + |B \setminus B_{\leqs 0.3/\sqrt{\gamma}}| < 0.2N = 0.4n$. Thus, the number of edges with at least one endpoint outside $A_{\leqs 0.3/\sqrt{\gamma}} \cup B_{\leqs 0.3/\sqrt{\gamma}}$ is less than $0.4Dn$. Due to this and our earlier bound on $|E_{\sat}|$, we have $|E_{\sat} \cap (A_{\leqs 0.3/\sqrt{\gamma}} \times B_{\leqs 0.3/\sqrt{\gamma}})| > 0.5Dn - 0.4Dn = 0.1Dn$. 

Each $(a, b) \in E_{\sat} \cap (A_{\leqs 0.3/\sqrt{\gamma}} \times B_{\leqs 0.3/\sqrt{\gamma}})$ is satisfied by $\phi$ with probability $\frac{1}{|\psi(a)||\psi(b)|} > 10\gamma$. Thus, the expected number of edges satisfied by $\phi$ is at least $10\gamma|E_{\sat} \cap (A_{\leqs 0.3/\sqrt{\gamma}} \times B_{\leqs 0.3/\sqrt{\gamma}})| > \gamma Dn$. In other words, $\Ex[\val_{\Pi}(\phi)] > \gamma$, which, as pointed out earlier, is a contradiction.
\end{proof}

The above proposition allows us to apply Chernoff bound to upper bound the probability that each multilabeling leaves few edges unsatisfied in $\Pi_{\interm}$, as stated more precisely below.

\begin{proposition} \label{prop:chernoff}
Let $\psi$ be any multilabeling of $\Pi$ such that $\val_{\Pi}(\psi) < 0.5$. Then, with probability at most $(2|\Sigma|)^{-10N/\sqrt{\gamma}}$, less than $0.2pD n$ edges in $E_{\interm}$ are unsatisfied by $\psi$. 
\end{proposition}

\begin{proof}
Let $E_{\unsat} \subseteq E$ denote the set of edges unsatisfied by $\psi$ in $\Pi$. Since $\val_{\Pi}(\psi) < 0.5$, we have $|E_{\unsat}| > 0.5Dn$. Since each edge $e \in E_{\unsat}$ is included to $E_{\interm}$ independently with probability $p$, we can apply Chernoff bound and conclude that the probability that $E_{\interm}$ contains at least $0.2pDn$ such edges is at most $\exp(-(0.6)^2(0.5pDn)) \leqs \exp(-10^{-5}\Delta n) \leqs (2|\Sigma|)^{-10N/\sqrt{\gamma}}$ as desired, where the last inequality follows from the choice of $\Delta$.
\end{proof}

With the above propositions in place, Lemma~\ref{lem:soundness} can now be proved by a simple union bound.

\begin{proof}[Proof of Lemma~\ref{lem:soundness}]
Suppose that $\val(\Pi) < \gamma$. Observe that, for any $t \in \N$, there is a one-to-one correspondence between multilabelings of $\Pi$ of cost $t$ and subsets of $(A \cup B) \times \Sigma$ of size $t$. As a result, the number of multilabelings of $\Pi$ of cost\footnote{Note that here we need $0.06N / \sqrt{\gamma} \leqs |\Sigma|N$. This is true because a random labeling of $\Pi$ satisfies $1/|\Sigma|$ fraction of edges in $\Pi$; hence, the assumption $\val(\Pi) < \gamma$ implies that $|\Sigma| > 1/\gamma$.} $\lfloor 0.06N / \sqrt{\gamma} \rfloor$ is 
\begin{align*}
\binom{N|\Sigma|}{\lfloor 0.06N / \sqrt{\gamma} \rfloor} 
\leqs \left(\frac{eN|\Sigma|}{\lfloor 0.06N / \sqrt{\gamma} \rfloor}\right)^{\lfloor 0.06 N / \sqrt{\gamma} \rfloor}
\leqs \left(\frac{eN|\Sigma|}{0.06N / \sqrt{\gamma}}\right)^{0.06 N / \sqrt{\gamma}}
\leqs \left(2|\Sigma|\right)^{0.5N / \sqrt{\gamma}}
\end{align*}
where the first inequality follows from the fact that $\binom{M}{K} \leqs (eM/K)^K$ for all $M, K \in \N$ and the second inequality is from the fact that $(eM/x)^x$ is an increasing function for $x \in [0, M)$.

Hence, by Proposition~\ref{prop:max-to-min}, Proposition~\ref{prop:chernoff} and union bound, we can conclude that, with probability at least $1 - (2|\Sigma|)^{-9.5N/\sqrt{\gamma}} > 0.99$, every multilabeling $\psi$ of cost $\lfloor 0.06N / \sqrt{\gamma} \rfloor$ leaves at least $0.2pDn$ edges in $E_{\interm}$ unsatisfied. Finally, recall from Proposition~\ref{prop:simple-facts} that, with probability 0.99, we have $|E_{\interm} \setminus E'| \leqs 0.1pDn$. Thus, with probability 0.98, no multilabeling $\psi$ of cost $\lfloor 0.06N / \sqrt{\gamma}\rfloor$ can satisfy all the edges in $E'$, which concludes our proof.
\end{proof}

\section{Discussions and Open Questions}

While our result implies almost optimal dependency between the gap and the degree for inapproximability of Min-Rep, getting an (almost) optimal dependency for alphabet size still remains a challenging open question. This is true not only for Min-Rep but also for Max-Rep. For Max-Rep, a random labeling satisfies at least $1/|\Sigma|$ fraction of edges in expectation, giving a $1/|\Sigma|$-approximation for the problem. A slightly better $\Omega(\log |\Sigma|/|\Sigma|)$-approximation is known~\cite{KKT16}; hence, the best trade-off one could hope for is $|\Sigma| = O(g \log g)$ where $g$ is the target gap. Again, we remark that the best known dependency is from Chan's work~\cite{Chan16} with $|\Sigma| = O(g^2\log g)$.

For Min-Rep, the situation is less clear. For instance, it is not even clear whether there is an approximation algorithm with ratio depending only on $|\Sigma|$. As a result, it may be interesting to study the three-way trade-off between $|\Sigma|$, the maximum degree $\Delta$ and the gap $g$ in this case. 

Another interesting question is to try to obtain hardness of approximation in terms of the number of variables $N = |A| + |B|$. This can be viewed as hardness in terms of $\Delta$ but when we allow $\Delta$ to be as large as $N$. To the best of our knowledge, no hardness of the form $N^{\delta}$ for some constant $\delta > 0$ is known for this regime, although this is believed by some to be true. (For instance, the projection games conjecture~\cite{Mosh15} implies such hardness.) Recently, Dinur and the author~\cite{DM18} proves $N^{1 - o(1)}$ ETH-hardness of approximation for a closely related 2-CSP problem, which is exactly the same is Max-Rep except that the constraints do not have to be projections. Unfortunately, the techniques there do not seem to directly apply to Label Cover.

\subsection*{Acknowledgment}

I would like to thank Bundit Laekhanukit for helpful discussions.

\bibliography{main}

\newcommand{\etalchar}[1]{$^{#1}$}
\begin{thebibliography}{ALM{\etalchar{+}}98}

\bibitem[AKS11]{AKS11}
Per Austrin, Subhash Khot, and Muli Safra.
\newblock Inapproximability of vertex cover and independent set in bounded
  degree graphs.
\newblock {\em Theory of Computing}, 7(1):27--43, 2011.

\bibitem[ALM{\etalchar{+}}98]{ALMSS}
Sanjeev Arora, Carsten Lund, Rajeev Motwani, Madhu Sudan, and Mario Szegedy.
\newblock Proof verification and the hardness of approximation problems.
\newblock {\em J. ACM}, 45(3):501--555, May 1998.

\bibitem[AS98]{AS}
Sanjeev Arora and Shmuel Safra.
\newblock Probabilistic checking of proofs: A new characterization of {NP}.
\newblock {\em J. ACM}, 45(1):70--122, January 1998.

\bibitem[CCK08]{CCK08}
Tanmoy Chakraborty, Julia Chuzhoy, and Sanjeev Khanna.
\newblock Network design for vertex connectivity.
\newblock In {\em STOC}, pages 167--176, 2008.

\bibitem[Cha16]{Chan16}
Siu~On Chan.
\newblock Approximation resistance from pairwise-independent subgroups.
\newblock {\em J. {ACM}}, 63(3):27:1--27:32, 2016.

\bibitem[CK08a]{CK08a}
Chandra Chekuri and Nitish Korula.
\newblock Single-sink network design with vertex connectivity requirements.
\newblock In {\em FSTTCS}, pages 131--142, 2008.

\bibitem[CK08b]{CK08}
Julia Chuzhoy and Sanjeev Khanna.
\newblock Algorithms for single-source vertex connectivity.
\newblock In {\em FOCS}, pages 105--114, 2008.

\bibitem[CK12]{CK12}
Julia Chuzhoy and Sanjeev Khanna.
\newblock An {$O(k^3 \log n)$}-approximation algorithm for vertex-connectivity
  survivable network design.
\newblock {\em Theory of Computing}, 8(1):401--413, 2012.

\bibitem[CK14]{CK14}
Chandra Chekuri and Nitish Korula.
\newblock A graph reduction step preserving element-connectivity and packing
  steiner trees and forests.
\newblock {\em {SIAM} J. Discrete Math.}, 28(2):577--597, 2014.

\bibitem[CLNV14]{CLNV14}
Joseph Cheriyan, Bundit Laekhanukit, Guyslain Naves, and Adrian Vetta.
\newblock Approximating rooted steiner networks.
\newblock {\em {ACM} Trans. Algorithms}, 11(2):8:1--8:22, 2014.

\bibitem[CMVZ16]{CMVZ16}
Julia Chuzhoy, Yury Makarychev, Aravindan Vijayaraghavan, and Yuan Zhou.
\newblock Approximation algorithms and hardness of the \emph{k}-route cut
  problem.
\newblock {\em {ACM} Trans. Algorithms}, 12(1):2:1--2:40, 2016.

\bibitem[DM18]{DM18}
Irit Dinur and Pasin Manurangsi.
\newblock {ETH}-hardness of approximating 2-{CSP}s and directed steiner
  network.
\newblock In {\em ITCS}, pages 36:1--36:20, 2018.

\bibitem[KKT16]{KKT16}
Guy Kindler, Alexandra Kolla, and Luca Trevisan.
\newblock Approximation of non-boolean {2CSP}.
\newblock In {\em SODA}, pages 1705--1714, 2016.

\bibitem[KS13]{KS13}
Subhash Khot and Muli Safra.
\newblock A two-prover one-round game with strong soundness.
\newblock {\em Theory of Computing}, 9:863--887, 2013.

\bibitem[Lae14]{Bun14}
Bundit Laekhanukit.
\newblock Parameters of two-prover-one-round game and the hardness of
  connectivity problems.
\newblock In {\em SODA}, pages 1626--1643, 2014.

\bibitem[Mos15]{Mosh15}
Dana Moshkovitz.
\newblock The projection games conjecture and the {NP}-hardness of ln
  n-approximating set-cover.
\newblock {\em Theory of Computing}, 11:221--235, 2015.

\bibitem[Nut09]{Nutov09}
Zeev Nutov.
\newblock A note on rooted survivable networks.
\newblock {\em Inf. Process. Lett.}, 109(19):1114--1119, 2009.

\bibitem[Nut12]{Nutov12}
Zeev Nutov.
\newblock Approximating minimum-cost connectivity problems via uncrossable
  bifamilies.
\newblock {\em {ACM} Trans. Algorithms}, 9(1):1:1--1:16, 2012.

\bibitem[Raz98]{Raz98}
Ran Raz.
\newblock A parallel repetition theorem.
\newblock {\em {SIAM} J. Comput.}, 27(3):763--803, 1998.

\end{thebibliography}
\bibliographystyle{alpha}

\end{document}